\documentclass{article}
\usepackage{spconf,graphicx}
\usepackage{algorithm}
\usepackage{stackengine}
\usepackage{caption}
\usepackage{cite}
\usepackage{color}
\usepackage{xcolor}
\usepackage{amsmath,amsthm,amssymb,amsfonts,bm}
\usepackage{textcomp}
\usepackage{multirow}
\usepackage{subcaption}
\usepackage[]{algpseudocode}

\makeatletter
\newcommand*{\rom}[1]{\expandafter\@slowromancap\romannumeral #1@}

\newtheorem{theorem}{Theorem}

\def \v x{\bm x}

\def \bs{\bm s}

\def \v x{\bm X}

\def \bSigma{\bm \Sigma}

\renewcommand{\v}[1]{\ensuremath{\boldsymbol{#1}}}
\DeclareMathOperator*{\tsum}{\textstyle\sum}

\newcommand{\recht}[1]{\operatorname{#1}}

\title{
Topology-dependent privacy bound for decentralized federated learning
}
\name{Qiongxiu Li$^{1}$, Wenrui Yu$^{2}$, Changlong Ji$^{3}$, Richard Heusdens$^{2,4}$ }

\address{$^1$Tsinghua University, China, qiongxiuli@mail.tsinghua.edu.cn\\
$^{2}$Delft University of Technology, the Netherlands, w.yu-6@student.tudelft.nl\\
$^3$Telecom SudParis, Institut Polytechnique de Paris,  France, changlong.ji@telecom-sudparis.eu\\
$^{4}$Netherlands Defence Academy, the Netherlands, r.heusdens@tudelft.nl
}

\begin{document}
\ninept
\maketitle
\raggedbottom
\addtolength{\abovedisplayskip}{-1.0mm}
\addtolength{\belowdisplayskip}{-1.0mm}
\begin{abstract}
Decentralized Federated Learning (FL) has attracted significant attention due to its enhanced robustness and scalability compared to its centralized counterpart. It pivots on peer-to-peer communication rather than depending on a central server for model aggregation. While prior research has delved into various factors of decentralized FL such as aggregation methods and privacy-preserving techniques, one crucial aspect affecting privacy is relatively unexplored: the underlying graph topology. In this paper, we fill the gap by deriving a stringent privacy bound for decentralized FL under the condition that the accuracy is not compromised, highlighting the pivotal role of graph topology. Specifically, we demonstrate that the minimum privacy loss at each model aggregation step is dependent on the size of what we term as 'honest components', the maximally connected subgraphs once all untrustworthy participants are excluded from the networks, which is closely tied to network robustness.   Our analysis suggests that attack-resilient networks will provide a superior privacy guarantee. We further validate this by studying both Poisson and power law networks, showing that the latter, being less robust against attacks, indeed reveals more privacy.
In addition to a theoretical analysis, we consolidate our findings by examining two distinct privacy attacks: membership inference and gradient inversion. 



\end{abstract}

\section{Introduction}
Federated Learning (FL) performs collaborative training between multiple participants/nodes/clients without directly sharing each node's raw data \cite{mcmahan2017communication}. It can be implemented in either a centralized/star topology or a decentralized topology. The centralized topology, which stands as the predominant topology in FL, employs a central server that interacts directly with each and every node. However, in real-world scenarios, maintaining such a centralized server can be challenging due to its high communication bandwidth demands and the requisite trust from all involved participants. Furthermore, centralized topologies possess an inherent vulnerability–a single point of failure-making them vulnerable to attacks aiming to bring down the entire network. As a remedy,  decentralized FL offers an alternative by substituting the server with distributed processing protocols which require information exchange between (locally) connected nodes only. Examples of these protocols are the empirical methods where the data aggregation is done using average consensus techniques such as gossiping SGD \cite{jin2016scale}, D-PSGD \cite{lian2017can} and variations thereof \cite{tang2018d,hu2019decentralized,jane2022gmm}.

Though FL does not require direct sharing of individual participants' private data, it is shown susceptible to privacy breaches.  The exchanged information, such as gradients or weights, can inadvertently lead to potential data leakages.  Most of the existing work focuses on the centralized FL, and examples of attacks include the membership inference attack\cite{shokri2017membership,li2023effective} and the gradient inversion attack \cite{zhu2019deep,geiping2020inverting,yin2021see,geng2023improved,yang2022using}. The goal of the membership inference attack is to determine whether a particular data point was used for training the target model (being a member) or not (being a non-member). It has been recently shown in \cite{li2023effective}  that membership information can be leaked through gradients by exploiting the so-called gradient orthogonality in data instances in overparameterized neural networks. The gradient inversion attack employs an iterative method to find fake data that produces a gradient similar to the real gradient generated by the private data. 
Such attacks are based on the assumption that data samples, when producing analogous gradients, are likely to be congruent. 
As a consequence, many approaches attempt to protect privacy by protecting the local gradients held by each node from being revealed to others.

Many studies claim that mere decentralization, i.e., deploying distributed processing protocols, would enhance the users' privacy. This is because no single user is as powerful as the centralized server, seemingly reducing privacy concerns \cite{cheng2019towards,vogels2021relaysum}. Yet, such an argument is challenged in a recent study \cite{pasquini2022privacy} which shows that untrustworthy users can influence other users' updates,  becoming as powerful as the central server in centralized FL. Consequently, it highlights the necessity of integrating cryptographic techniques, such as differential privacy (DP) \cite{dwork2006} and secure aggregation (SA)\cite{bonawitz2017practical},  to further enhance the privacy of decentralized FL.   DP methods, including the LEASGD approach \cite{cheng2019towards} and ADMM-based approaches \cite{zhang2016dynamic,zhang2018improving,zhang2018recycled,huang2019dp,zhang2022privacy}, pose an inherent trade-off between accuracy and privacy as employing DP will inevitably lead to a compromise in accuracy.  The SA approaches, on the other hand, do not deteriorate the accuracy but require high communication overhead. 

Decentralized FL's privacy has been analyzed via various aspects including different types of distributed tools, potential threats from dishonest users, and privacy-enhancing methods. However, an often overlooked yet critical factor is the underlying graph topology. Although \cite{pasquini2022privacy} provides some insights indicating that a denser graph  can reduce privacy risks, it's not just the density, but also the degree distribution that is pivotal. For instance, two graphs with identical density can exhibit vastly different privacy implications (as we will show later). We believe that a thorough investigation of topology is crucial in shaping decentralized FL's privacy, requiring an in-depth exploration.  In this paper, we bridge this gap by analytically establishing a privacy bound on decentralized FL, highlighting the tight connection between graph topology and privacy. This newly derived bound is especially significant for real-world applications as it provides guidance on which topological structures intrinsically enhance privacy. Our main results are summarized as follows:


\begin{enumerate}
    \item For decentralized FL that guarantees the output accuracy of model aggregation is uncompromised,  we derive a bound on the privacy loss, highlighting  that privacy leakage is profoundly influenced by the underlying graph topology, or more precisely, the network robustness against attacks.
    \item Our findings suggest that robust network topologies, which sustain connectivity under adversarial conditions, inherently offer better privacy protection. We confirm this by investigating two prevalent  topologies: Poisson and power law networks, with the latter being less privacy-preserving. Our results are substantiated through two distinct privacy attack assessments: membership inference and gradient inversion.
\end{enumerate}

\section{Preliminaries}
This section reviews the necessary fundamentals for the paper. 
\subsection{Centralized federated learning}
FL generally considers a centralized setting assuming there is a centralized server connected to a number of users/clients/nodes.  
Assuming there are $n$ clients and denote $\v w^{(t)}$ as the model weights at iteration $t$. A typical FL protocol works as follows: 
\begin{enumerate}
    \item Initialization: at iteration  $t=0$, the central server randomly initializes the weights $\v w^{(0)}$ of the global model.
    \item Local model training: at each iteration $t$, each user $i$ first receives the model updates from the server and then calculates its local gradient, denoted as $\v g_i^{(t)}$,  based on one mini-batch sampled from its local dataset.
    \item Model aggregation: the server gathers local gradients from users and aggregates them to obtain an updated global model. The aggregation is often done by weighted averaging and  typically uniform weights are applied, i.e., 
\begin{align}\label{eq.gave}  
\v g_{\rm{ave}}^{(t)}=\frac{1}{n} \sum_{i=1}^{n} \v g_i^{(t)} 
\end{align}
After obtaining the aggregated gradient, each node $i$ then updates its own model weight  by $\v  w_i^{(t+1)}=\v w_i^{(t)}-\eta \v g_{\rm{ave}}^{(t)}$, where $\eta$ is a constant. 
The last two steps  are repeated  until the global model converges or a certain stopping criterion is met. 
\end{enumerate}

\subsection{Decentralized federated learning}
Decentralized FL works for cases where a trusted centralized server is not available. Many decentralized FL protocols work by deploying distributed average consensus algorithms to compute the average of local gradients, i.e., computing \eqref{eq.gave} without any centralized coordination. Examples are gossip  \cite{dimakis2010gossip}, linear iterations \cite{olshevsky2009convergence}, and convex optimization-based methods such as the ADMM \cite{boyd2011distributed} and the PDMM \cite{zhang2018distributed}, which all rely on peer-to-peer communication over distributed networks.  A distributed network is often modelled as an undirect graph $\mathcal{G}=(\mathcal{V},\mathcal{E})$ where $\mathcal{V}={\{1,2,...,n}\}$ and $\mathcal{E}\subseteq \mathcal{V}\times \mathcal{V}$ denote the set of nodes and edges, respectively.  
Note that node $i$ can only communicate with (neighboring) node $j$ if $(i,j)\in \mathcal{E}$.

We now take the linear iteration algorithm \cite{olshevsky2009convergence} as an example to explain how to conduct peer-to-peer aggregation when computing \eqref{eq.gave}.  Let $\v x^{(0)}$ be the so-called state vector of the network which is initialized with local gradients, i.e., 
\begin{align}\label{eq.linearIni}
\v x^{(0)} =\v g,    
\end{align}
where $\v g=(\v g^\top_1,\v g^\top_2,\ldots, \v g^\top_n)^\top$ \footnote{To avoid confusion, the superscript $^{(t)}$ is omitted here.}  is the vector of stacked local gradients of all nodes. The average gradient can be obtained by applying, at every step $r = 1,2,\ldots,r_{\max}$, a linear transformation $\v A\in \mathcal{A}$ where 
\begin{align}
    \mathcal{A}=\left\{\v A \in \mathbb{R}^{n \times n} \,|\, \v A_{i j}=0 \text { if }(i,j) \notin \mathcal{E} \text { and } i \neq j\right\},
\end{align}
such that the state vector $\v x$ is updated as
\begin{align}
 \v x^{(r+1)}=\v A \v x^{(r)},   
\end{align}
and  the optimum solution is $\forall i\in \mathcal{V}:~\v x_i^*=\v g_{\rm{ave}}$. The structure of $\v A$ reflects the connectivity of the network\footnote{For simplicity, we assume that $\v A$ is constant for all iterations, representing a synchronous execution of the algorithm. For asynchronous implementation, the transformation depends on which node will update. The results presented here can be readily generalized to asynchronous cases by considering expected values. }.
In order to correctly compute the average, that is, $\v x^{(r)}\rightarrow  n^{-1}\bm 1 \bm 1^{\top}\v g$ as $r\rightarrow\infty$ where $\bm 1 \in \mathbb{R}^{n}$ denote the vector of all ones, necessary and sufficient conditions for $\v A$ are (i) $\mathbf{1}^{\top}\v A=\mathbf{1}^{\top}$, (ii)  $\v A \mathbf{1}=\mathbf{1}$, (iii) $\alpha\left(\v A-\frac{\bm1\bm1^{\top}}{n}\right)<1$,
where $\alpha(\cdot)$ denotes the spectral radius \cite{olshevsky2009convergence}.

Hence, decentralized FL often shares the same updating procedure as the centralized case except at each model aggregation step for computing \eqref{eq.gave}, where a peer-to-peer distributed algorithm is applied to average the local gradients.  
\subsection{Adversary model}
We consider a commonly utilized passive (also known as honest-but-curious)  adversary model for distributed systems.  The passive adversary works by colluding a number of nodes in the network, termed as corrupt nodes. These nodes  can share information together to enhance the chance of inferring the private information of the remaining nodes, referred to as honest nodes.




\section{Privacy bound and graph topology}
The privacy leakage depends on many factors, for example the type of learning protocols,  the deployed privacy-preserving techniques, and the amount of corrupt nodes. However, another pivotal factor that has been overlooked is the underlying \emph{graph topology}.  As we shall see,  the privacy of honest nodes depends on the graph topology, or more precisely, on the sizes of the honest components  (maximally connected subgraphs formed by honest nodes) it belongs to after removing all corrupt nodes.  
\subsection{Minimum information loss} \label{sec:honest}
Let $\mathcal{V}_h$ and $\mathcal{V}_{\textrm{c}}=\mathcal{V}\setminus \mathcal{V}_h$ denote the set of honest and corrupt nodes in the network, respectively. 
\begin{theorem}\label{theo:df}
Let $\mathcal{G}_h=({\mathcal{V}}_h,\mathcal{E}_h)$ be the subgraph of $\mathcal{G}$ after all corrupt nodes are eliminated. 
Let $\mathcal{G}_{ h,1},\ldots,\mathcal{G}_{ h,k_h}$ denote the components of $\mathcal{G}_h$ and let ${\mathcal{V}}_{ h,k}$ be the vertex set of $\mathcal{G}_{ h,k}$.
Given any protocol that outputs $f(s_1,s_2,\ldots,s_n)=\tsum_{i=1}^n s_i$ to every node, after the protocol,  the  partial sums
\begin{align} \label{eq.psum}
    \{ \tsum_{i\in {\mathcal{V}}_{ h,k}}s_i\}_{ k=1,2,\ldots,k_h},
\end{align}
will  be revealed to the adversary.
\end{theorem}
\begin{proof}
We use the case  $k_h=2$ to explain the main idea, i.e., the honest nodes are disconnected into two components after removing the corrupt nodes. Denote $\mathcal{L}$ and $\mathcal{R}$ as the node set of two honest components and let $\v{s}_\mathcal{L}$, $\v{s}_\mathcal{R}$ and $\v{s}_{{\mathcal{V}}_c}$ be vectors consisting of the inputs for two honest components, and the corrupt nodes, respectively.
Denote $\operatorname{F}$ as the protocol outputting $\sum_{i=1}^n s_i$. The adversary has the following view
  \begin{align}\label{eq:view_sum_proof}
\{\v{s}_{{\mathcal{V}}_c},\v{r}_{{\mathcal{V}}_c},\v{m}_\mathcal{L},\v{m}_\mathcal{R}, f(\v{s}_\mathcal{L},\v{s}_{{\mathcal{V}}_c},\v{s}_\mathcal{R})\},
  \end{align}
where $\v{r}_{{\mathcal{V}}_c}$ is a vector containing the so-called randomness from the corrupt nodes, $\v{m}_\mathcal{L}$ is a vector containing all messages received from nodes in $\mathcal{L}$ and similarly for $\v{m}_\mathcal{R}$  in $\mathcal{R}$. 

The adversary can emulate a protocol execution by determining the actions for the nodes in $\mathcal{R}$. This involves selecting certain inputs and randomness and then proceeding with the steps in $\operatorname{F}$, using the messages $\v{m}_\mathcal{L}$ as needed in $\operatorname{F}$. If there is a discrepancy between an entry in $\v{m}_\mathcal{R}$ and the view, the protocol will be terminated and restarted. Since $\v{s}_\mathcal{R}$ and $\v{r}_\mathcal{R}$ are valid choices (though there may be several others), the adversary will eventually succeed, implying that it will identify $\tilde{\v{s}}_\mathcal{R}$ and $\tilde{\v{r}}_\mathcal{R}$ that provide the exact information view as in equation (\ref{eq:view_sum_proof}). Given that the adversary utilizes the identical messages from the nodes in $\mathcal{L}$ as in the actual execution of $\operatorname{F}$, and the accurate output is included in the view, we must have  $f(\v{s}_\mathcal{L},\v{s}_{{\mathcal{V}}_c},\v{s}_\mathcal{R})=f(\v{s}_\mathcal{L},\v{s}_{{\mathcal{V}}_c},\tilde{\v{s}}_\mathcal{R})$ and thus the adversary can determine
  \begin{align*}
    \tsum_{i\in \mathcal{L}} s_i = f(\v{s}_\mathcal{L},\v{s}_{{\mathcal{V}}_c},\v{s}_\mathcal{R})-\tsum_{i \in {\mathcal{V}}_c} s_i -\tsum_{i \in R}\tilde{s}_i,
  \end{align*}
 after knowing $\tsum_{i \in \mathcal{L}}s_i$  the adversary can also determine $\tsum_{i \in \mathcal{R}}s_i=\tsum_{i=1}^n s_i-\tsum_{i \in {\mathcal{V}}_c}s_i-\tsum_{i \in \mathcal{L}}s_i$.  The above proof can easily be extended to the case where $k_h>2$ and in such cases the adversary will learn all the sums of the private data in each component. Hence, the proof is now complete.
\end{proof}


\subsection{Lower bound of privacy}
To develop an understanding on the minimum information loss,  we use an information-theoretical measure, e.g., mutual information to quantify the privacy leakage. Let $\recht{I}(\cdot;\cdot)$ and $\recht{h}(\cdot)$ denote mutual information and differential entropy, respectively \cite{cover2012elements}.   Suppose that the $\bs_i$s are realizations of an i.i.d.\ random vector $\{S_1,\ldots, S_n\}$, and denote the variable $\bSigma_k = \sum_{i \in {\mathcal{V}}_{ h,k}} \!\!S_i$. Consider an honest node $i\in \mathcal{V}_h$, given the knowledge of \eqref{eq.psum} the adversary will learn the following information about its private data $s_i$:
\begin{align}\label{eq.bound}
\recht{I}(S_i;\bSigma_1,\bSigma_2,\ldots,\bSigma_{k_h}).
\end{align}
Note that the above bound is consistent with the results presented in \cite{kreitz2010practical,beimel2007private,li2019privacyA,Jane2020TSP,Jane2020TIFS}, where it was shown that the partial sums of honest components are revealed when computing average/sum. Hence, the derived lower bound is tight.

Denote $m_{k}=|{\mathcal{V}}_{ h,k}|, k\in \{1,2,\ldots,k_h\}$ as the number of nodes within the respective honest components.
To see how these sizes affect the privacy leakage, denote the variance of $S_i$s as ${\rm var}(S_i) = \sigma^2$. For sufficiently large $m_k$, the variable $\bSigma_k $ is approximately Gaussian distributed with  variance ${\rm var}(\bSigma_k) =m_k\sigma^2$;  Consider an honest node $i\in \mathcal{V}_h$,  let $\mathcal{G}_{h,1}$ denote the component that node $i$ belongs to. Thus, \eqref{eq.bound} becomes
\begin{align}
\recht{I}(S_i;\bSigma_1)&\stackrel{(a)}{=} \recht{h}(\bSigma_1) - \recht{h}(\bSigma_1|S_i)\nonumber\\
& \stackrel{(b)}{\approx} \frac{1}{2}\left(\log(2\pi\textrm{e}m_1\sigma^2)-\log(2\pi\textrm{e}(m_1-1)\sigma^2)\right) \nonumber\\
& = \frac{1}{2}\log\left(\frac{m_1}{m_1-1}\right)\nonumber\\
&\approx \frac{1}{2(m_1-1)}  \label{eq.miSigma1},
\end{align}
where \eqref{eq.bound} equals $\recht{I}(S_i;\bSigma_1)$ uses the fact that $S_i$ is independent of all $\bSigma_k$s except of $\bSigma_1$; (a) uses the definition of mutual information\footnote{Shannon entropy $H(\cdot)$ can be used for discrete random variables.}; (b) assumes that $m_1$ is sufficiently large and uses the fact that the differential entropy of a Gaussian distribution with variance $\sigma^2$ is given by $\frac{1}{2} \log (2 \pi \textrm{e} \sigma^2 )$.

Overall, we conclude that the privacy loss of an honest node is approximately inversely proportional to the number of nodes in the honest component it belongs to. That means, \emph{the bigger the size of the honest component is, the less the privacy loss is}.  This is verified in Fig.~\ref{fig:mi} where we depict the mutual information  $\recht{I}(S_i;\bSigma_1)$ as a function of the honest component size $m_1$ considering two distributions of $S_i$s: Gaussian distributed with zero mean and unit variance, and uniformly distributed over the interval $[0,1]$. For each experiment, we conducted 5000 Monto Carlo runs  and adopted  the non-parametric entropy estimation toolbox \cite{ver2000non} to estimate the mutual information.  
\begin{figure}[t]
\centering
\includegraphics[width=.35\textwidth]{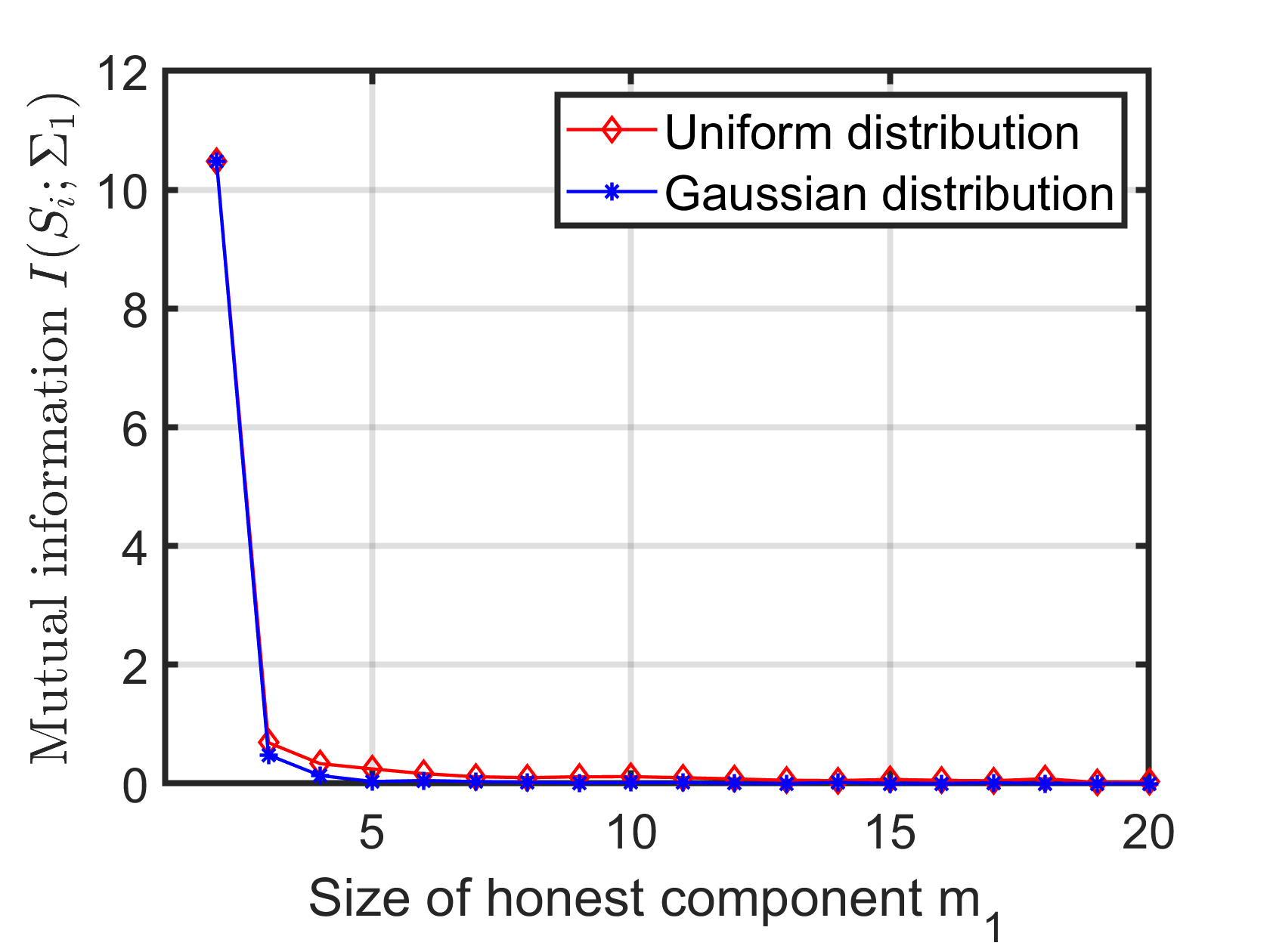}
 \centering
\caption{Mutual information $\recht{I}(S_i;\bSigma_1)$ of as a function of the honest component size $m_1$ for Gaussian and uniform distributed private data $S_i$.}
\label{fig:mi}
\vskip -10pt
\end{figure}

\subsection{Graph topology and network robustness}
In Theorem \ref{theo:df} and the subsequent result detailed in \eqref{eq.miSigma1}, we demonstrated that privacy loss is dependent on the size of honest components.  This concept closely aligns with graph robustness or resilience when considering an adversary colluding a number of nodes in the network. The topology of a graph plays a pivotal role in understanding the robustness of networks\cite{magnien2011impact}. One fundamental aspect of graph topology is the degree distribution, which characterizes the number of connections each node in the network possesses. The literature  predominantly investigates two key types of degree distribution: Poisson and power law distributions.
\begin{figure}[t]
\centering
\includegraphics[width=.30\textwidth]{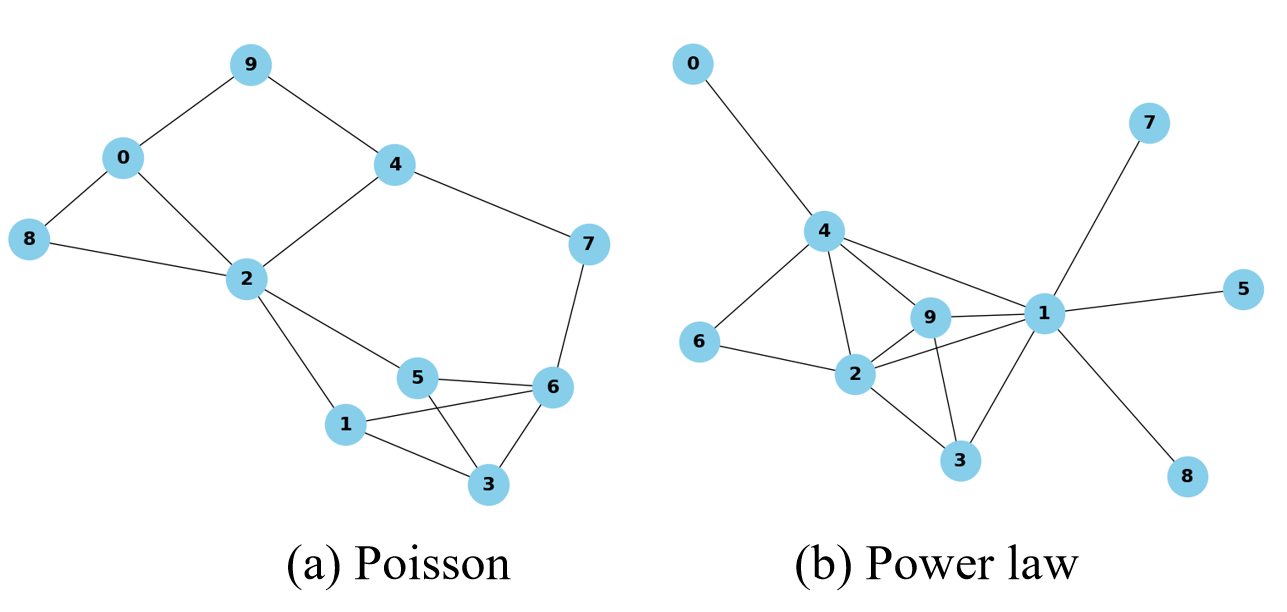}
 \centering
\caption{Sample network with  $n=10$ nodes of (a) Poisson and (b) power law distribution.}
\label{fig:example}
\vskip -10pt
\end{figure}

Many networks, such as those resembling random structures, follow a Poisson distribution [Fig.~\ref{fig:example} (a)], i.e., the degree of each node follows a Poisson distribution with a specific average degree. These networks exhibit a relatively uniform degree distribution, with most nodes having degrees close to the average. Conversely, many real-world networks, like the Internet and social or biological networks, exhibit power law degree distributions. In these power law graphs, only a few nodes have extremely high degrees, whereas most have significantly fewer connections  [Fig.~\ref{fig:example} (b)]. 
Assuming the adversary is aware of the graph's connectivity, it will target and corrupt the pivotal nodes, i.e., those with high degrees. Given that power law networks rely on a few highly connected nodes to ensure network connectivity, they are notably 
 more vulnerable to such targeted attacks compared to Poisson networks \cite{magnien2011impact}. Such increased vulnerability can lead to greater privacy risks, as we'll discuss in the following section. 

\section{Experimental validations}
We now proceed to consolidate our theoretical results via numerical validations. Recall that peer-to-peer model aggregation \eqref{eq.gave} is applied at all iterations, given the fact that  the revealed information at each model aggregation step is the sum of local gradients of honest components (Theorem \ref{theo:df}), we conclude that throughout the whole process the adversary can collect the following information:
\begin{align} \label{eq.gsum}
    \{\tsum_{i\in {\mathcal{V}}_{ h,k}}\v g_i^{(t)}\}_{t\geq 0, k=1,2,\ldots,k_h}.
\end{align}
To quantitatively evaluate the privacy leakage caused by the above gradient information, we deploy two widely-used attacks including membership inference and gradient inversion.  As we shall see, the privacy loss depends on the graph topology, or more precisely, the size of honest components after removal of corrupted nodes. 
\begin{figure}[t]
\centering
\includegraphics[width=.30\textwidth]{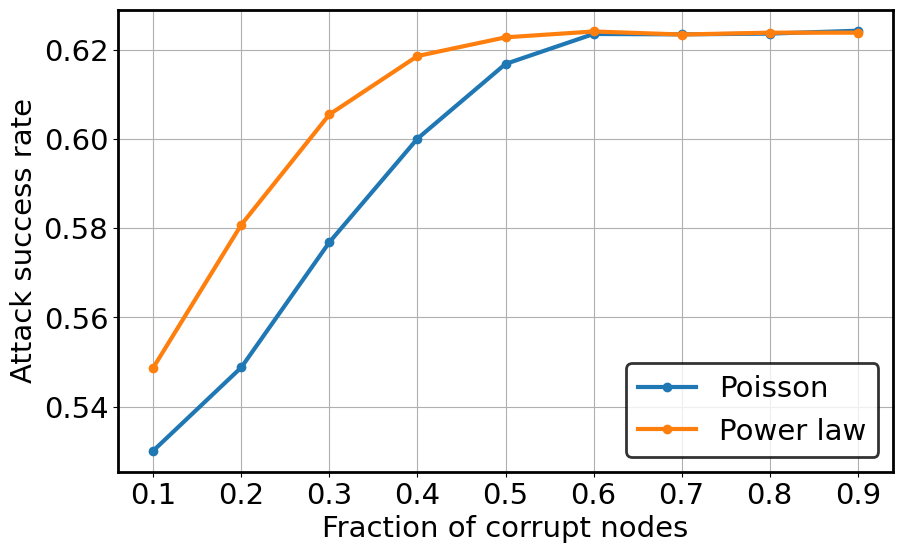}
 \centering
\caption{Membership privacy leakage of the honest nodes, i.e., overall attack success rate,  as a function of the fraction of corrupt nodes for both Poisson and power law networks.}
\label{fig:mia}
\vskip -10pt
\end{figure}
\begin{figure}[t]
\centering
\includegraphics[width=.3\textwidth]{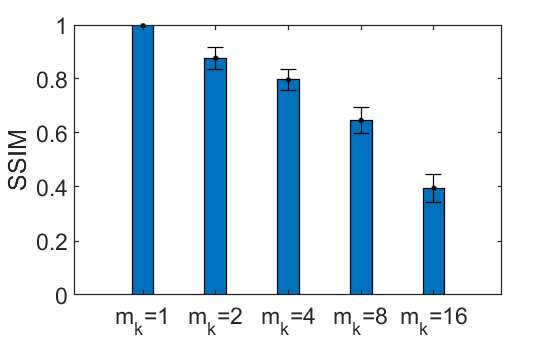}
 \centering
 \mbox{}\vspace*{-.8\baselineskip}
\caption{Image quality comparison of the reconstructed inputs via inverting gradients using structural similarity index measure (SSIM) of different sizes of honest component $m_k=1,~2,~4,~8,~16$.}
\label{fig:ssim}
\vskip -10pt
\end{figure}

\begin{figure}[t]
\centering
\includegraphics[width=.40\textwidth]{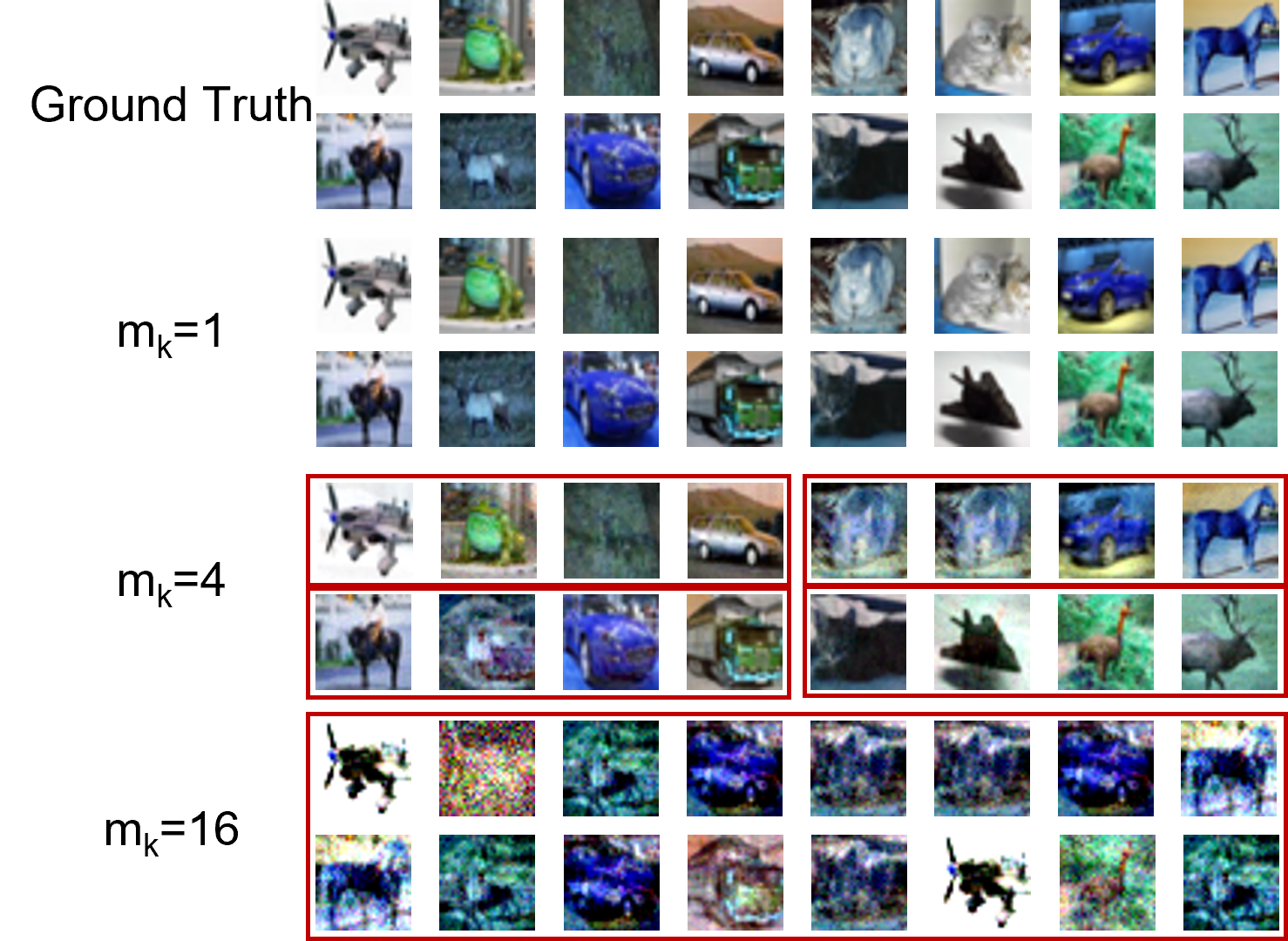}
 \centering
 \mbox{}\vspace*{-.8\baselineskip}
\caption{Samples images of input reconstruction using gradient inversion attack under three different sizes of honest component $m_k=1,~4,~16$, wherein the red box indicates that the corresponding samples are from the same component.}
\label{fig:gia}
\vskip -10pt
\end{figure}

\subsection{Membership inference attack} 
 With the gradient information in  \eqref{eq.gsum} we deploy the recent attack proposed in \cite{li2023effective} as it is specifically designed to infer  membership information from gradients.
We first examine the relationship between privacy loss and the type of graph topology. We randomly generated Poisson and power law networks with 10 nodes and 15 edges (see Fig \ref{fig:example} for examples).   We employed the CIFAR-10 dataset, dividing it into 10 nodes, each containing 4000 data samples, and trained it using the AlexNet model. The test performance of this decentralized FL protocol is similar to the result reported in \cite{li2023effective}. In Fig.~\ref{fig:mia} we demonstrate the attack success rate as a function of the fraction of corrupt nodes for both topologies, where the results are averaged over 5000 Monte Carlo runs.  Note that the adversary is strategic and has prior knowledge about the graph topology, thus it will corrupt the pivotal nodes first. To simulate such scenario,  we assume that nodes are sequentially removed in descending order based on their degrees.   Clearly, while in both cases the membership privacy leakage increases as the number of corrupt nodes increases, the power law networks reveal more privacy compared to Poisson networks. Hence, power law networks are more vulnerable to targeted attacks and thus pose a higher privacy risk compared to Poisson networks. 

\subsection{Gradient inversion attack}
We now further investigate the size of the honest component and  privacy leakage, we deploy the existing gradient inversion attack \cite{zhu2019deep} to invert input samples from the observed sum of local gradients in each honest component. We simulated a randomly connected network with $n=50$ nodes and the batchsize for generating the local gradient is set to one.  In Fig.~\ref{fig:ssim} shows the quality of the reconstructed images for different sizes of honest components: $m_k=1,2,4,8,16$. As expected, the reconstruction quality degrades with increasing size of the honest components.  To visualize the results, in Fig.~\ref{fig:gia} we further demonstrate some examples of the reconstructed images (due to space limit only three cases $m_k = 1, 4, 16$ are shown). Hence, we can see that as the size of honest component $m_k$ increases, there is an obvious degradation in the reconstruction quality. Especially for the case that there are images in the same component with the same label, e.g., the two cats in the top right of the case $m_k=4$, the reconstruction performance is poor. This is consistent with the common observation that gradient inversion attack does not perform well in the case of repeated labels \cite{geng2023improved}.  Overall, we conclude that the bigger the size of honest component, the less accurate will the reconstructed input be.

\section{Conclusion}
In this paper we emphasized the pivotal role of graph topology in the privacy of Decentralized FL. By establishing a delicate privacy bound for model aggregation, we unveiled the tight relationship between network structure and privacy risks. Through our exploration of Poisson and power law networks, we determined that certain topologies inherently offer better privacy guarantees when dealing with attacks.  Our findings, validated by practical experiments including membership inference and gradient inversion, lay a foundation for advancements in the privacy of decentralized FL systems.

\newpage
\bibliographystyle{IEEEbib}
\bibliography{dualpath}

\begin{thebibliography}{10}

\bibitem{mcmahan2017communication}
B.~McMahan, E.~Moore, D.~Ramage, S.~Hampson, and B.~A. y~Arcas,
\newblock ``Communication-efficient learning of deep networks from decentralized data,''
\newblock in {\em Proc. Int. Conf. Artif. Intell. Statist.} PMLR, pp. 1273--1282, 2017.

\bibitem{jin2016scale}
P.~H. Jin, Q.~Yuan, F.~Iandola, and K.~Keutzer,
\newblock ``How to scale distributed deep learning?,''
\newblock {\em arXiv:1611.04581}, 2016.

\bibitem{lian2017can}
{X. Lian, C. Zhang, H. Zhang, C. Hsieh, W. Zhang and J. Liu},
\newblock ``Can decentralized algorithms outperform centralized algorithms? a case study for decentralized parallel stochastic gradient descent,''
\newblock in {\em Proc. Adv. Neural Inf. Process. Syst.}, vol.~30, 2017.

\bibitem{tang2018d}
H.~Tang, X.~Lian, M.~Yan, C.~Zhang, and J.~Liu,
\newblock ``D$^2$: Decentralized training over decentralized data,''
\newblock in {\em Proc. Int. Conf. Mach. Learn.} PMLR, pp. 4848--4856, 2018.

\bibitem{hu2019decentralized}
C.~Hu, J.~Jiang, and Z.~Wang,
\newblock ``Decentralized federated learning: A segmented gossip approach,''
\newblock {\em arXiv:1908.07782}, 2019.

\bibitem{jane2022gmm}
Q.~Li, J.~K. Gundersen, K.~Tjell, R.~Wisniewski, and M.~G. Christensen,
\newblock ``Privacy-preserving distributed expectation maximization for gaussian mixture model using subspace perturbation,''
\newblock in {\em Proc. Int. Conf. Acoust., Speech, Signal Process.} IEEE, pp. 4263--4267, 2022.

\bibitem{shokri2017membership}
{R. Shokri, M. Stronati, C. Song, and V. Shmatikov},
\newblock ``Membership inference attacks against machine learning models,''
\newblock in {\em Proc. IEEE Symp. Secur. Privacy}, pp. 3--18, 2017.

\bibitem{li2023effective}
J.~Li, N.~Li, and B.~Ribeiro,
\newblock ``Effective passive membership inference attacks in federated learning against overparameterized models,''
\newblock in {\em Proc. Int. Conf. Learn. Represent.}, 2023.

\bibitem{zhu2019deep}
L.~Zhu, Z.~Liu, and S.~Han,
\newblock ``Deep leakage from gradients,''
\newblock in {\em Proc. Adv. Neural Inf. Process. Syst.}, vol.~32, 2019.

\bibitem{geiping2020inverting}
J.~Geiping. H.~Bauermeister, H.~Dr{\"o}ge and M.~Moeller,
\newblock ``Inverting gradients-how easy is it to break privacy in federated learning?,''
\newblock in {\em Proc. Adv. Neural Inf. Process. Syst.}, vol.~33, pp. 16937--16947, 2020.

\bibitem{yin2021see}
H.~Yin, A.~Mallya, A.~Vahdat, JM~Alvarez, J.~Kautz, and P.~Molchanov,
\newblock ``See through gradients: Image batch recovery via gradinversion,''
\newblock in {\em Proc. IEEE Conf. Comput. Vis. Pattern Recognit.}, pp. 16337--16346, 2021.

\bibitem{geng2023improved}
J.~Geng, Y.~Mou, Q.~Li, F.~Li, O.~Beyan, S.~Decker, and C.~Rong,
\newblock ``Improved gradient inversion attacks and defenses in federated learning,''
\newblock {\em IEEE Trans. Big Data}, 2023.

\bibitem{yang2022using}
H.~Yang, M.~Ge, K.~Xiang, and J.~Li,
\newblock ``Using highly compressed gradients in federated learning for data reconstruction attacks,''
\newblock {\em IEEE Trans. Inf. Forensics Secur.}, vol. 18, pp. 818--830, 2022.

\bibitem{cheng2019towards}
H.~Cheng, P.~Yu, H.~Hu, S.~Zawad, F.~Yan, S.~Li, H.~Li, and Y.~Chen,
\newblock ``Towards decentralized deep learning with differential privacy,''
\newblock in {\em Proc. Int. Conf. Cloud Comput.} Springer, pp. 130--145, 2019.

\bibitem{vogels2021relaysum}
T.~Vogels, L.~He, A.~Koloskova, S.~P. Karimireddy, T.~Lin, S.~U. Stich, and M.~Jaggi,
\newblock ``Relaysum for decentralized deep learning on heterogeneous data,''
\newblock in {\em Proc. Adv. Neural Inf. Process. Syst.}, vol.~34, pp. 28004--28015, 2021.

\bibitem{pasquini2022privacy}
D.~Pasquini, M.~Raynal, and C.~Troncoso,
\newblock ``On the privacy of decentralized machine learning,''
\newblock {\em arXiv:2205.08443}, 2022.

\bibitem{dwork2006}
{C. Dwork},
\newblock ``Differential privacy,''
\newblock in {\em Proc. Int. Colloq. Automata, Languages, Program.}, pp. 1--12, 2006.

\bibitem{bonawitz2017practical}
K.~Bonawitz, V.~Ivanov, B.~Kreuter, A.~Marcedone, H.~B. McMahan, S.~Patel, D.~Ramage, A.~Segal, and K.~Seth,
\newblock ``Practical secure aggregation for privacy-preserving machine learning,''
\newblock in {\em Proc. ACM SIGSAC Conf. Comput. Commun. Secur.}, pp. 1175--1191, 2017.

\bibitem{zhang2016dynamic}
{T. Zhang and Q. Zhu},
\newblock ``Dynamic differential privacy for {ADMM}-based distributed classification learning,''
\newblock {\em IEEE Trans. Inf. Forensics Secur.}, vol. 12, no. 1, pp. 172--187, 2016.

\bibitem{zhang2018improving}
{X. Zhang, M. M. Khalili, and M. Liu},
\newblock ``Improving the privacy and accuracy of {ADMM}-based distributed algorithms,''
\newblock in {\em Proc. Int. Conf. Mach. Lear.}, pp. 5796--5805, 2018.

\bibitem{zhang2018recycled}
{X. Zhang, M. M. Khalili, and M. Liu},
\newblock ``Recycled {ADMM}: Improve privacy and accuracy with less computation in distributed algorithms,''
\newblock in {\em Proc. 56th Annu. Allerton Conf. Commun., Control, Comput.}, pp. 959--965, 2018.

\bibitem{huang2019dp}
{Z. Huang, R. Hu, Y. Gong, and E. Chan-Tin},
\newblock ``{DP-ADMM}: {ADMM}-based distributed learning with differential privacy,''
\newblock {\em IEEE Trans. Inf. Forensics Secur.}, vol. 15, pp. 1002--1012, 2019.

\bibitem{zhang2022privacy}
Z.~Zhang, S.~Yang, W.~Xu, and K.~Di,
\newblock ``Privacy-preserving distributed admm with event-triggered communication,''
\newblock {\em IEEE Trans. Neural Netw. Learn, Syst.}, 2022.

\bibitem{dimakis2010gossip}
{A. G. Dimakis, S. Kar, J. M. Moura, M. G. Rabbat, and A. Scaglione},
\newblock ``Gossip algorithms for distributed signal processing,''
\newblock {\em Proc. IEEE}, vol. 98, no. 11, pp. 1847--1864, 2010.

\bibitem{olshevsky2009convergence}
{A. Olshevsky and J. Tsitsiklis},
\newblock ``Convergence speed in distributed consensus and averaging,''
\newblock {\em SIAM J. Control Optim.}, vol. 48, no. 1, pp. 33--55, 2009.

\bibitem{boyd2011distributed}
{S. Boyd, N. Parikh, E. Chu, B. Peleato, J. Eckstein, et~al.},
\newblock ``Distributed optimization and statistical learning via the alternating direction method of multipliers,''
\newblock {\em Found. Trends in Mach. Learn.}, vol. 3, no. 1, pp. 1--122, 2011.

\bibitem{zhang2018distributed}
{G. Zhang and R. Heusdens},
\newblock ``Distributed optimization using the primal-dual method of multipliers,''
\newblock {\em IEEE Trans. Signal Process.}, vol. 4, no. 1, pp. 173--187, 2018.

\bibitem{cover2012elements}
{T. M. Cover and J. A. Tomas},
\newblock {\em Elements of information theory},
\newblock John Wiley \& Sons, 2012.

\bibitem{kreitz2010practical}
G.~Kreitz, M.~Dam, and D.~Wikstr{\"o}m,
\newblock ``Practical private information aggregation in large networks,''
\newblock in {\em Inf. Secur. Technolo. Appl.} Springer, pp. 89--103, 2012.

\bibitem{beimel2007private}
{A. Beimel},
\newblock ``On private computation in incomplete networks,''
\newblock {\em Distrib. Comput.}, vol. 19, no. 3, pp. 237--252, 2007.

\bibitem{li2019privacyA}
{Q. Li, I. Cascudo, and M. G. Christensen},
\newblock ``Privacy-preserving distributed average consensus based on additive secret sharing,''
\newblock in {\em Proc. Eur. Signal Process. Conf.}, pp. 1--5, 2019.

\bibitem{Jane2020TSP}
{Q. Li, R. Heusdens and M. G. Christensen},
\newblock ``Privacy-preserving distributed optimization via subspace perturbation: A general framework,''
\newblock in {\em IEEE Trans. Signal Process., vol. 68, pp. 5983 - 5996}, 2020.

\bibitem{Jane2020TIFS}
{Q. Li, J. S. Gundersen, R. Heusdens and M. G. Christensen},
\newblock ``Privacy-preserving distributed processing: Metrics, bounds, and algorithms,''
\newblock {\em IEEE Trans. Inf. Forensics Secur.}, vol. 16, pp. 2090--2103, 2021.

\bibitem{ver2000non}
G.~Ver Steeg,
\newblock ``Non-parametric entropy estimation toolbox (npeet),''
\newblock {\em https://github.com/gregversteeg/NPEET}, 2000.

\bibitem{magnien2011impact}
C.~Magnien, M.~Latapy, and J.~L. Guillaume,
\newblock ``Impact of random failures and attacks on poisson and power-law random networks,''
\newblock {\em ACM Comput. Surv.}, 2011.

\end{thebibliography}

\end{document}